\documentclass[11pt]{article}
\pdfoutput=1
\usepackage{fullpage}
\usepackage{hyperref}
\usepackage{float}
\usepackage{epsfig}
\usepackage{amsmath}
\usepackage{amssymb}
\usepackage{amstext}
\usepackage{amsmath}
\usepackage{xspace}
\usepackage{latexsym}
\usepackage{verbatim}
\usepackage{multirow}
\usepackage{ifthen}
\usepackage{url}

\usepackage{times}
\newtheorem{theorem}{Theorem}[section]
\newtheorem{definition}{Definition}

\newtheorem{lemma}[theorem]{Lemma}

\newtheorem{corollary}[theorem]{Corollary}

\newcommand{\qed}{\mbox{\ \ \ }\rule{6pt}{7pt} \bigskip}

\renewcommand{\comment}[1]{}
\newenvironment{proof}{\noindent{\em Proof:}}{\hfill\qed}

\newenvironment{proofof}[1]{\noindent{\em Proof of #1.}}{\qed}

\newcommand{\A}{{\mathcal A}}
\newcommand{\OPT}{{\mbox {OPT}}}

\newcommand{\iid}{\text{i.i.d.}}
\newcommand{\wrt}{\text{w.r.t.}}

\newcommand{\val}{v}
\newcommand{\vals}{{\mathbf \val}}
\newcommand{\valsmi}{{\mathbf \val}_{-i}}
\newcommand{\vali}[1][i]{{\val_{#1}}}
\newcommand{\valith}[1][i]{{\val_{(#1)}}}

\newcommand{\eff}{e}

\newcommand{\effi}[1][i]{{\eff_{#1}}}

\newcommand{\bid}{b}

\newcommand{\virt}{\phi}

\newcommand{\ivirt}{\bar{\virt}}

\newcommand{\dist}{F}

\newcommand{\dens}{f}

\newcommand{\densi}[1][i]{{\dens_{#1}}}

\newcommand{\price}{p}
\newcommand{\prices}{{\mathbf \price}}
\newcommand{\pricei}[1][i]{{\price_{#1}}}

\newcommand{\alloc}{x}
\newcommand{\allocs}{{\mathbf \alloc}}
\newcommand{\alloci}[1][i]{{\alloc_{#1}}}

\newcommand{\util}{u}

\newcommand{\utili}[1][i]{{\util_{#1}}}

\newcommand{\prob}[2][]{\text{\bf Pr}\ifthenelse{\not\equal{}{#1}}{_{#1}}{}\!\left[#2\right]}
\newcommand{\expect}[2][]{\text{\bf E}\ifthenelse{\not\equal{}{#1}}{_{#1}}{}\!\left[#2\right]}
\newcommand{\given}{\,\mid\,}

\newcommand{\Rev}[1]{{\text{\bf Rev}}[#1]}
\newcommand{\MP}[1]{{\text{\bf MP}}[#1]}
\newcommand{\MPi}[1]{{\text{\bf MP}_i}[#1]}

\newcommand{\virta}{\psi}
\newcommand{\ivirta}{\overline{\virta}}

\newcommand{\stat}{a}

\begin{document}
\title{Optimal Crowdsourcing Contests}

\author{Shuchi Chawla\thanks{Computer Sciences Dept., University of Wisconsin -
  Madison. Supported in part by NSF award 
  CCF-0830494 and in part by a Sloan Foundation fellowship.\tt{ shuchi@cs.wisc.edu}.}
\and Jason D. Hartline\thanks{Electrical Engineering and Computer Science, Northwestern
   University. Supported in part 
   by NSF award CCF-0830773.\tt{ hartline@eecs.northwestern.edu}.}
 \and Balasubramanian Sivan\thanks{Computer Sciences Dept., University
 of Wisconsin - Madison.  Supported in part by NSF award 
  CCF-0830494. \tt{ balu2901@cs.wisc.edu}.}
}
\date{}
\maketitle{}

\thispagestyle{empty}

\begin{abstract}
We study the design and approximation of optimal crowdsourcing
contests.  Crowdsourcing contests can be modeled as all-pay auctions
because entrants must exert effort up-front to enter.  Unlike all-pay
auctions where a usual design objective would be to maximize revenue,
in crowdsourcing contests, the principal only benefits from the
submission with the highest quality.  We give a theory for optimal
crowdsourcing contests that mirrors the theory of optimal auction
design: the optimal crowdsourcing contest is a virtual valuation
optimizer (the virtual valuation function depends on the distribution
of contestant skills and the number of contestants).  We also compare
crowdsourcing contests with more conventional means of procurement.
In this comparison, crowdsourcing contests are relatively
disadvantaged because the effort of losing contestants is wasted.
Nonetheless, we show that crowdsourcing contests are 2-approximations
to conventional methods for a large family of ``regular''
distributions, and $4$-approximations, otherwise.

\end{abstract}

\newpage
\setcounter{page}{1}

\section{Introduction}
Crowdsourcing contests have become increasingly important and
prevalent with the ubiquity of the Internet.  For instance, instead of
hiring a research team to develop a better collaborative filtering
algorithm, Netflix issued the ``Netflix challenge'' offering a million
dollars to the team that develops an algorithm that beats the Neflix
algorithm by 10\%.  More generally, Taskcn allows users to post tasks
with monetary rewards, collects submissions by other users, and
rewards the best submission; and many Q\&A sites allow users to post
questions and reward the best answer with much-coveted ``points''.  We
address two questions in this paper, (a) what format of crowdsourcing
competition induces the highest-quality winning contribution, and (b)
how inefficient is crowdsourcing over more conventional means of
contracting.

Crowdsourcing competitions can be modeled as {\em all-pay} auctions.
In the highest-bid-wins single-item all-pay auction, the auctioneer
solicits payments (as bids), awards the item to the agent with the
highest payment, and keeps all the agent payments.  These auctions are
well understood in settings where each agent has an independent
private value for obtaining the item.  In the connection to
crowdsourcing contests, the ``item'' is the monetary reward, the
payments are the submissions, and the private value is the rate at
which the contestant works.  However, unlike all-pay auctions, in
crowdsourcing competitions the principal usually only values the
winning submission and has no value for lesser submissions.
Therefore, while the performance metric for auctions is usually {\em
revenue} which is the sum of the agent payments, in crowdsourcing
contests where payments are submissions, the relevant performance
metric is the quality of the best submission, i.e., the maximum agent
``payment''.

The {\em revenue equivalence} principle implies that in equilibrium
the revenue of the highest-bid-wins all-pay auction is the same as
that of first- and second-price auction formats; however, in these
latter auction formats only the winner makes a payment.  Since
non-winners make payments in all-pay auctions, the maximum agent
payment in all-pay auctions is lower than that of first- and
second-price auctions.  To connect this auction theory back to the
setting of procurement, first- and second-price auctions are analogous
to conventional procurement mechanisms, e.g., for government
contracts, whereas the all-pay format is analogous to crowdsourcing
contests.  Importantly, the performance metric for first- and
second-price procurement auctions is their revenue, that is, the
winner's payment. While the all-pay auction obtains the same total
revenue, the principal in crowdsourcing cannot attain this full
revenue and therefore suffers a loss relative to conventional methods.

Our first result is to show that in expectation the maximum agent
payment in highest-bid-wins all-pay auctions is at least half its
total revenue.  Consequently crowdsourcing contests can extract from
the best submission at least half of the total contribution from the
crowd, which in turn implies that they are $2$-approximate with
respect to conventional procurement via highest-bid-wins auctions.
Of course, highest-bid-wins auctions are not necessarily revenue
optimal. However, for a large class of distributions (termed {\em
regular}), auctions that award the item to the highest bidder that
meets a reservation price are optimal \cite{mye-81}.  In these
settings crowdsourcing contests that require submissions to be of a
minimum quality (e.g., the Netflix challenge required submissions to
beat the Netflix algorithm by 10\%) are $2$-approximations.  For more
general distributional settings reserve pricing gives a
$2$-approximation to the optimal auction revenue~\cite{CHMS-10} and
crowdsourcing contests with minimum quality conditions are, therefore,
$4$-approximations to conventional procurement.
This approximation 
also implies a ``simple versus optimal'' style result, i.e., that the gains
from precisely optimizing a contest based on the distribution versus running a
simple highest-bid-wins crowdsourcing contest with a minimum quality
condition are at most a factor of $4$.

Our second result derives the optimal static crowdsourcing format that
maximizes the quality of the best submission. Specifically, suppose we
fix the reward for the $k$-th best submission to be $\stat_k$ (where
we normalize the $\stat_k$'s to $\sum_{k=1}^{n}\stat_k = 1$). What
should the $\stat_k$'s be? For instance on computer programming
crowdsourcing site TopCoder.com, the best submission receives 2/3rds
and the second-best submission receives 1/3rd of the total reward.  Is
this a better format than awarding the entire amount to the best
submission in terms of the quality of that submission?  We prove that
it is not: $\stat_1 = 1$ and $\stat_k = 0$ for $k > 1$, or
winner-takes-all, is the optimal choice over all such static contests.

Of course in some settings it may be better to adjust the number of
rewards and their distribution across the participants dynamically as
functions of the observed submission qualities. Our third result
derives the format of crowdsourcing contests that dynamically
optimizes the quality of the best submission. We give a complete
characterization of optimal crowdsourcing contests. In what would be a
familiar result to auction theorists, optimal crowdsourcing
contests are ``ironed virtual value optimizers'' in that the reward is
divided evenly among all contestants whose submissions are tied under
a weakly monotone transformation (via the ironed virtual value
function) of the submission quality.  Importantly, the number of
contestants who share the reward is determined dynamically and each
contestant's share is the same.  Perhaps surprisingly, and unlike the
case of classical auction theory, the transformation to ironed virtual
values depends on the number of contestants.

Optimal crowdsourcing contests require the auctioneer to know the
distribution of agents' skills, e.g. in order to pick an appropriate
minimum submission quality. In our fourth result we consider the loss
from not knowing the distribution. For the revenue objective, Bulow
and Klemperer~\cite{BK-96} proved that for regular distributions
recruiting an extra bidder is more profitable to the auctioneer than
knowing the distribution. We show that this result implies that a
simple highest-bid-wins contest approximates the optimal contest
within a factor approaching $2$; this limits the benefit of knowing
the skill distribution.

\paragraph{Related Work.}

This paper follows the connection made between crowdsourcing contests
and all-pay auctions from DiPalantino and Vojnovic \cite{dPV09} and
questions from Archak and Sundararajan \cite{AS09} and Moldovanu and
Sela~\cite{MS01,MS06} on optimizing the reward structure to improve
the quality of the best submissions. \cite{AS09} and \cite{MS01}
compare winner-take-all crowdsourcing contests against ones with a
statically determined division of the reward among top agents, e.g.,
as in the TopCoder.com mechanism. The objective in \cite{MS01} is the
sum of the qualities of submissions (analogous to revenue in our
discussion) and the Archak-Sundararajan objective is the cumulative
effort from the top $k$ agents less the monetary reward. Both papers
show that when agents' submission qualities are linear in their effort
winner-take-all is optimal over other static divisions. The
Moldovanu-Sela result also holds when quality is a convex function of
effort, but not generally for concave functions.  Minor~\cite{M11}
studies a generalization of the problem of Moldovanu and Sela, and
derives the optimal crowdsourcing contest via a Myersonian ironing
approach (again, to maximize the sum of qualities).

Moldovanu and Sela~\cite{MS06} study both the highest quality
submission objective and the revenue (sum of all submission qualities)
objective with the total reward normalized to 1. 
They compare the performance of two-stage contests against
one-stage contests. Among one-stage contests they consider both a
single grand contest, as well as many sub-contests in parallel, with
the winner of each sub-contest receiving a prize. In two-stage
contests, the winners of the first stage sub-contests compete in a
final round. These are all static contests in the sense that the
division of reward among winners of different sub-contests is
predetermined. For the sum-of-qualities objective, \cite{MS06} prove
that a single grand contest is best among these contest formats. For
the highest quality objective, if there are sufficiently many
competitors then it is optimal to split the competitors in two
divisions and to have a final among the two divisional winners.
Further, as the number of competitors tends to infinity, ~\cite{MS06}
show that the optimal highest quality objective is at least half of
the optimal sum-of-qualities objective --- we generalize this result
and show that the factor of two ratio holds for any number of
competitors.

In this paper our goal is to optimize the quality of the best
submission (unlike~\cite{MS01,M11} which consider total quality) with
a total reward normalized to one (unlike \cite{AS09} which optimizes
the quality of the best submissions less the monetary reward). We
study optimal crowdsourcing contests over all single-stage all-pay
formats, unlike \cite{MS06} which limits the format of one-stage
contests but also studies two-stage contests. Our main results are to
show that the wasted effort is not large and to characterize optimal
crowdsourcing contests that can potentially divide the award between
agents dynamically depending on the qualities of submissions. We also
show that our model is consistent with that of \cite{MS01,AS09} in
that the optimal static allocation of the award is
winner-take-all\footnote{Moldovanu and Sela in \cite{MS06} also state
that this should be true but do not provide a reference or a proof.}.

The following other results relating to crowdsourcing contests are
technically unrelated to ours.  DiPalantino and Vojnovic~\cite{dPV09}
study crowdsourcing websites as a matching market.  They discuss
equilibria where contestants first choose which contest to participate
in and then their level of effort.  Yang et al.~\cite{YAA08} and
DiPalantino and Vojnovic~\cite{dPV09} empirically study bidder behavior
from crowdsourcing website Taskcn and conclude that experienced
contestants strategize better than others and their strategizes match
the BNE predictions fairly well.  

There have been a number of studies of all-pay auctions in complete
information settings (e.g., Baye et al.~\cite{BKdV96}), but these works
are also technically unrelated to ours.

\section{Preliminaries}

%
%

\paragraph{Auction Theory.}
Consider the standard auction-theoretic problem of selling a single
item to $n$ agents.  Each agent $i$ has a private value $\vali$ for
receiving the object and is risk-neutral with linear utility $\utili =
\vali \alloci - \pricei$ for receiving the item with probability
$\alloci$ and making payment $\pricei$. An auction $\A$
solicits bids and determines the outcome which consists of an
allocation $\allocs = (\alloci[1],\ldots,\alloci[n])$ and payments
$\prices = (\pricei[1],\ldots,\pricei[n])$.

%
%
Suppose that the agents' values are drawn $\iid$ from continuous
distribution $\dist$ (that is, having no point-masses) with
distribution function $\dist(z) = \prob{\vali \leq z}$ and density
function $\dens(z)$. A Bayes-Nash equilibrium (BNE) in auction $\A$ is
a profile of strategies for mapping values to bids in the auction that
are a mutual best response, i.e., when the values are drawn from
$\dist$ and other agents follow their equilibrium strategies then each
agent (weakly) prefers to also follow the prescribed strategy over
taking any other action. 

%
%
Formally, on valuation profile $\vals = (\vali[1], \ldots,\vali[n])$,
denote the composition of an auction and a strategy profile by
allocation rule $\allocs(\vals)$ and payment rule
$\prices(\vals)$. When agent $i$ is bidding in the auction, she knows
her own value $\vali$ and assumes that the other agent values are
drawn from the distribution $\dist$.  Denote her {\em interim}
allocation and payment rules as $\alloci(\vali) =
\expect[\vals]{\alloci(\vals) \given \vali}$ and $\pricei(\vali) =
\expect[\vals]{\pricei(\vals) \given \vali}$, respectively.
Bayes-Nash equilibrium requires that $\vali \alloci(\vali) -
\pricei(\vali) \geq \vali\alloci(z) - \pricei(z)$ for all $z$ and from
this constraint is derived the standard characterization of BNE:
\begin{theorem} \cite{mye-81}
\label{t:BNE-char}
Allocation and payment rules $\allocs(\cdot)$ and $\prices(\cdot)$ are
in BNE if and only if for all $i$
\begin{enumerate}
\item $\alloci(\vali)$ is monotone non-decreasing in $\vali$ and
\item $\pricei(\vali) = \vali \alloci(\vali) - \int_0^{\vali} \alloci(z)dz + \pricei(0)$
\end{enumerate}
where usually $\pricei(0) = 0$.
\end{theorem}
A simple consequence of this characterization is the {\em revenue
  equivalence} principle which states that two mechanisms with the
  same equilibrium allocation have the same equilibrium revenue---in
  fact each agent's expected interim payment is the same.

%
%
There are three standard formats for highest-bid-wins single-item
auctions: first-price, second-price, and all-pay.  In the first-price
variant the highest bidder wins and pays her bid, in the second-price
variant (a.k.a.~the Vickrey auction) the highest bidder wins and pays the
second highest bid, and in the all-pay variant the highest bidder wins
and all bidders pay their bids.  These auction formats all have BNE in
which the agent with the highest valuation wins;
Therefore, revenue equivalence implies that they have the same
expected revenue (sum of payments) in equilibrium.

%
%
The highest-bid-win auction formats do not always yield the highest
expected revenue.  To solve for optimal auctions,
Myerson~\cite{mye-81} defined {\em virtual valuations for revenue} as
$\virt(\vali) = \vali - \frac{1-\dist(\vali)}{\dens(\vali)}$ and
proved that the expected payment of an agent,
$\expect[\vali]{\pricei(\vali)}$ is equal to her expected virtual
value $\expect[\vali]{\virt(\vali)\alloci(\vali)}$.  The distribution
$\dist$ is said to be {\em regular} if the virtual valuation function
is monotone.  For regular distributions, maximizing virtual values
point-wise is a monotone allocation rule, and therefore can be
implemented in BNE.  The corresponding revenue-optimal auction serves
the agent with the highest positive virtual value.  By symmetry, this
agent is identically the agent with the highest value that meets a
reserve price of $\virt^{-1}(0)$.
\begin{theorem} 
\label{t:regular-reserve-optimal}
\cite{mye-81} When the virtual valuation function $\virt(\cdot)$ is
monotone, the optimal auction format is highest-bid-wins with a
reservation value of $\virt^{-1}(0)$, a.k.a., the {\em monopoly price}.
\end{theorem}

%
%
It will be useful to be able to solve for the equilibrium strategies
in all-pay auctions with reserves.  Revenue equivalence makes this
easy: the expected payment of an agent with value $\vali$ is the same
in both the all-pay and the second-price auction formats.  Of course in
the all-pay format the agent always pays her bid; therefore, her bid
$\bid(\vali)$ must be equal to her expected payment in the
second-price auction.  Let $\valith[j]$ denote the $j$th largest
value.  Agent $i$'s expected payment in the second-price auction when $\vali \geq r$,
is
exactly $\expect[\valsmi]{\max(\valith[2],r) \given \vali
= \valith[1]} \prob[\valsmi]{\vali = \valith[1]}$, so her bid in the
all-pay auction must be equal to this expectation.  

\begin{lemma} \label{l:all-pay-bids}
In a highest-bid-wins all-pay auction with value reserve $r$ an agent with
value $\vali$ bids
$$\bid(\vali)
= \expect[\valsmi]{\max(\valith[2],r) \given \vali = \valith[1]} \prob[\valsmi]{\vali
= \valith[1]}.$$
if $\vali\ge r$ and $0$ otherwise.
\end{lemma}

%
%
The reserve specified above is in value-space.  To implement such a
reserve in an auction, one must translate it to a reserve in
bid-space.  For first- and second-price auctions this transformation
is the identity function.  For all-pay auctions, it can be calculated
as follows.  An agent with value equal to the reserve $r$ in the
second price auction pays the reserve if she wins, i.e., her expected
payment is $r \prob[\valsmi]{r = \valith[1]} = r \dist(r)^{n-1}$.  By revenue
equivalence the same agent in the equivalent all-pay auction must bid
this expected payment; as this bid is the minimum bid that should be
accepted, it is the reserve.
\begin{lemma}
\label{l:all-pay-reserves}
The highest-bid-win all-pay auction with reserve bid
$r \dist(r)^{n-1}$ implements the highest-value-wins allocation rule
with a reserve value of $r$.
\end{lemma}

%
%
For irregular distributions, i.e., when $\virt(\cdot)$ is
non-monotone, the revenue-optimal auction is not reserve-price based.
Instead it selects the highest virtual value subject to monotonicity
of the allocation rule.  This optimization can be simplified by a very
general {\em ironing} technique.
\begin{theorem} 
\label{t:irregular-ivv-optimal}
\cite{mye-81,HR-08} 
There is an {\em ironing procedure} that converts any virtual
valuation function $\virt(\cdot)$ to a {\em ironed virtual valuation}
function $\ivirt(\cdot)$ that is monotone and has the property that
maximizing $\virt(\cdot)$ subject to monotonicity (of the allocation
rule) is equivalent to maximizing $\ivirt(\cdot)$ point-wise, with
ties broken randomly.  The BNE with this outcome is optimal.
\end{theorem}

%
%
\paragraph{Crowdsourcing.}
The following model for crowdsourcing contests and its connection to
all-pay auctions was proposed in~\cite{dPV09}.  To outsource a task to
the crowd a principal announces a monetary reward (normalized to 1).
Each of $n$ agents (the crowd) enters a submission.  Agent $i$'s {\em
skill} is denoted by $\vali$ and with {\em effort}, $\effi$, she can
produce a submission with {\em quality} $\pricei = \vali \effi$, i.e.,
her skill can be thought of as a rate of work and her effort the
amount of work.  Each agent's skill is her private information.  If
$\alloci$ fraction of the reward is awarded to agent $i$ then her
utility is $\utili = \alloci - \effi$.  From her perspective $\vali$
is a constant so maximizing utility is equivalent to maximizing
$\vali\utili = \vali\alloci - \pricei$.  Notice that this latter
formulation of the agent's objective mirrors that from the single-item
auction setting discussed previously; furthermore, as the agents exert
effort up-front, crowdsourcing contests intrinsically have all-pay
semantics.  Because of this connection, it will be convenient to refer
interchangeably to contests as auctions, skills as values, submission qualities as payments,
and rewards as allocations.

%
%
The objective for crowdsourcing contests $\A$ is to maximize the
quality of the best submission.  Because of the connection to all-pay
auctions we refer to this objective as the {\em maximum payment}
objective and denote its value for an auction $\A$ as $\MP{\A} =
\expect[\vals]{\max_i \pricei(\vals)}$.  This objective is quite
different from the standard revenue maximization objective $\Rev{\A} =
\expect[\vals]{\sum_i \pricei(\vals)}$.

%
%
One aim of this paper is to quantify the loss the principal incurs from
running an all-pay auction versus a more conventional means of
contracting.  For instance, standard formats for procurement auctions
are first- or second-price.  Importantly, in first- and second-price
auctions $\A$ all the payment comes from the highest bidder, therefore
$\MP{\A} = \Rev{\A}$ and the principal is able to extract quality
workmanship with no loss.  In contrast, in all-pay auctions which are
revenue equivalent to first- and second-price auctions the maximum
payment is not equal to the total revenue and thus the efforts of
non-winners constitute a loss in performance.  We thus quantify the
{\em utilization ratio} of an auction $\A$ as
$\frac{\Rev{\A}}{\MP{\A}}$.

%
%
We will see that the all-pay auction that optimizes maximum payment is
not the same as the auction (all-pay or otherwise) that maximizes
revenue.  We define the {\em approximation ratio} of an all-pay
auction to quantify its maximum payment relative to the revenue of
the optimal (first-price or second-price) auction, i.e., $\A$'s approximation ratio is
$\frac{\Rev{\OPT}}{\MP{\A}}$.  The {\em cost of crowdsourcing} (over
conventional procurement) is then the approximation ratio of the best
all-pay auction, i.e., $\inf_{\A} \frac{\Rev{\OPT}}{\MP{\A}}$.

\paragraph{Non-zero density Assumption.}
For the rest of this paper, the space of valuations $V$ is assumed to be an interval
and the density function $\dens(\cdot)$ is assumed to be
non-zero everywhere in $V$.

\section{Utilization and approximation ratios}
\label{sec:poa}

As noted previously, the maximum payment of a second- or first-price
auction is equal to its total payment or revenue. On the other hand,
in all-pay auctions the payment made by non-winners leads to a loss in
performance. In this section we quantify this loss for a special class
of all-pay auctions, namely those that always reward the highest
bidder subject to an anonymous reserve price. This further allows us
to find a simple all-pay auction that approximates optimal
procurement.

In this section we consider highest-bidder-wins reserve-price auctions
under either second-price or all-pay semantics. It is easy to see that
under all-pay semantics these auctions induce symmetric continuous
increasing bid functions at BNE and therefore their allocation
function is identical to a second-price auction with an appropriate
reserve price.

\begin{theorem}
\label{thm:util-ratio}
Let $\A$ be any highest-bidder-wins reserve-price all-pay
auction. Then $\Rev{\A}\le 2\MP{\A}$. That is, its utilization ratio
is bounded by $2$.
\end{theorem}
\begin{proof}
Let $\allocs$ denote the allocation function of the auction and
suppose that the bid function that it induces in BNE is given by
$\bid(\val)$.  We can write the expected revenue of the auction as the
sum of the contribution from the winning agent (i.e. the agent with
the maximum payment), and the contribution from other agents. Call the
first term $A$ and the second $B$. 
\begin{align*}
\Rev{\A} &= \underbrace{\sum_i \int_{\val} \bid(\val) \prob[\valsmi]{\val=\valith[1]}\dens(\val)\,d\val}_A
+ \underbrace{\sum_i \int_{\val} \bid(\val) (1-\prob[\valsmi]{\val=\valith[1]})\dens(\val)\,d\val}_B
\end{align*}
Note that $A$ is precisely $\MP{\A}$. We will now show that $A\ge B$,
or $A-B\ge 0$. By the revenue equivalence principle, $\bid(\val)$ is
equal to the expected payment that an agent with value $\val$ makes in
a second-price auction with the same allocation rule
(Lemma~\ref{l:all-pay-bids}). Let $g(\val)$ denote the expected payment 
in the second-price auction with reserve, given that $\val$ is the highest
value. Then we get that $\bid(\val) =
g(\val)\prob[\valsmi]{\val=\valith[1]}$. We note that $\prob[\valsmi]{\val=\valith[1]} =
\dist(v)^{n-1}$ is a strictly increasing function since $\dens(v) \neq
0$ for all $v \in V$. Now we can write $A-B$ as
\begin{align*}
A - B
&= \sum_i \int_{\val} \bid(\val)(2\prob[\valsmi]{\val=\valith[1]}-1)\dens(\val)\,d\val\\
&= \sum_i \int_{\val} g(\val)\dist(\val)^{n-1}(2\dist(\val)^{n-1}-1)\dens(\val)\,d\val\\
&= n\cdot \int_{\val} g(\dist^{-1}(t))t^{n-1}(2t^{n-1}-1)\,dt
\end{align*}
where, in the third equality, we substituted $t$ for $F(v)$. 

Next we note that ignoring the $g$ term, the integral is non-negative:
\begin{align*}
\int_{0}^{1}t^{n-1}(2t^{n-1}-1)\,dt
= \frac{2}{2n-1} - \frac{1}{n} > 0
\end{align*}
Let us consider the effect of the $g$ term. The function
$t^{n-1}(2t^{n-1}-1)$ vanishes for two values of
$t$ namely $0$ and $(1/2)^{\frac{1}{n-1}}$. Between these
two values the function is negative, and for $t >
(1/2)^{\frac{1}{n-1}}$, the function is positive. So when the function
is multiplied by $g(\dist^{-1}(t))$, a non-decreasing function of $t$, the negative
portion of the integral is magnified to a smaller extent than the
positive portion, implying that the integral stays positive. This
completes the proof.
\end{proof}

\paragraph{Tightness of Theorem~\ref{thm:util-ratio}.}
We now exhibit an example where the utilization ratio of $2$ is
tight. Consider a setting with $n$ agents, with each agent's value
distributed independently according to the $U[0,1]$
distribution. Consider the second-price auction with no reserve price.
The expected revenue of this auction can be computed to be
$\frac{n-1}{n+1}$. The corresponding all-pay auction induces a bid function
$\bid(\val) = \frac{n-1}{n}\val^n$. The expected revenue of the all-pay auction is the same as
that of the second-price auction, namely, $\frac{n-1}{n+1}$, which
approaches $1$ as $n$ increases. On the other hand, the expected
maximum payment can be computed to be
$\frac{n-1}{2n}$ which approaches $1/2$ as $n$
increases. In Section~\ref{sec:opt} we will revisit this example and
show that even the optimal all-pay auction (which is slightly better)
only achieves an expected maximum payment approaching $1/2$ for this
setting.

\paragraph{Utilization ratio for other all-pay auctions.}
We note that the bound on utilization ratio does not hold for
arbitrary symmetric all-pay auctions. For example, the all-pay auction
corresponding to a revenue-optimal auction that requires ironing over
large intervals of values induces a bidding function that is constant
over those intervals. This results in many agents being tied for the
reward, all making the same (low) payments but only one contributing
to the maximum payment.

\paragraph{Approximation ratio.} Recall that the approximation ratio
of an all-pay auction $\A$ is $\frac{\Rev{\OPT}}{\MP{\A}}$, where
$\OPT$ is the revenue optimal auction.  We now use the bound on
utilization ratio to prove that all-pay auctions achieve good
approximation ratios. In particular, we note that for regular
distributions highest-bidder-wins reserve-price auctions are
revenue-optimal (Theorem~\ref{t:regular-reserve-optimal}). For
irregular distributions, \cite{CHMS-10} show that highest-bidder-wins
auctions with an anonymous reserve price are within a factor of $2$ of
optimal\footnote{While highest-bidder-wins auctions with a
  non-anonymous reserve give a better approximation to the optimal
  revenue, they induce an asymmetric all-pay auction and
  Theorem~\ref{thm:util-ratio} does not apply.}. The following
corollaries then follow from Theorem~\ref{thm:util-ratio} upon
applying the revenue equivalence principle.

\begin{corollary}
\label{cor:reg-approx}
When agents' value distributions are regular, there exists an $\alpha$
such that the highest-bid-wins all-pay auction with reserve bid $\alpha$ 
achieves an approximation ratio at most $2$.
\end{corollary}
\begin{corollary}
\label{cor:nonreg-approx}
For all $\iid$ value distributions, there exists an $\alpha$
such that the highest-bid-wins all-pay auction with reserve bid 
$\alpha$ achieves an approximation ratio at most $4$.
\end{corollary}

These corollaries imply that the \emph{cost of crowdsourcing} is
always small --- no more than $4$.  The above example with uniform
distributions shows that the approximation factor in
Corollary~\ref{cor:reg-approx} is tight. An extension of the same
example in Section~\ref{sec:opt} shows that the worst-case cost of
crowdsourcing can be no smaller than $2$.

\section{Optimal crowdsourcing contests}
\label{sec:opt}

In this section we characterize optimal crowdsourcing contests, first
over a limited class of so-called ``static'' contests, and then over
all contests.

\paragraph{Static Contests.}Consider the class of contests that
predetermine the division of the reward into $\stat_1,\dots \stat_n$,
with $\sum_i \stat_i =1$. Agents are ordered by their submission
qualities and awarded the corresponding fraction of reward, i.e., the
$i$th best submission gets an $\stat_i$ fraction of the reward.  Note
that the Topcoder.com example mentioned in the introduction, where the
best submission receives 2/3rds and the second-best submission
receives 1/3rd of the total reward, falls under this class of
contests. For this class, the following theorem shows that the optimal
contest allocates the entire reward to the best submission; we defer
the proof to Appendix~\ref{sec:app}.

\begin{theorem}\label{thm:statOpt}
When the bidders' valuations are $\iid$, the optimal static
all-pay auction is a highest-bid-wins auction.

\end{theorem}

\paragraph{Symmetric Contests.} For the rest of this section, we focus on the class
of arbitrary symmetric contests. 
A symmetric
auction is one where a permutation of bids results in
the same permutation of the allocation and payments. Since agents'
private values are identically distributed, any such symmetric
allocation rule induces a symmetric equilibrium in which all agents
use an identical bidding function. This in turn implies that the
allocation as a function of agents' values is also symmetric across
agents.

We first present a characterization of the expected maximum payment of
any symmetric all-pay auction in terms of an appropriately defined
virtual value function. This characterization immediately implies that
the optimal mechanism is a virtual value maximizer.

\begin{definition}
\label{def:vv}
For a given distribution $\dist$ with density function $\dens$ and an
integer $n$, we define the {\em virtual value for maximum payment},
$\virta_n(z)$ as
\[\virta_n(z) = z\dist(z)^{n-1} - \frac{1-\dist(z)^{n}}{n\dens(z)}\]
\end{definition}

\begin{lemma}
\label{lem:mp=vv}
Consider a setting with $n$ agents and values distributed \iid\ 
according to distribution $\dist$. Let $\A$ be a symmetric all-pay
auction implementing the allocation function $\allocs$. Then $\MP{\A}
= \expect{\sum_i \alloci(\vals)\virta_n(\vali)}$.
\end{lemma}

\begin{proof}
Suppose that the allocation function $\allocs$ induces a symmetric bid
function $\bid(\cdot)$ on the agents. Recall that by the revenue
equivalence principle, $\bid(\val)$ is equal to the expected payment
that an agent with value $\val$ makes under $\allocs(\cdot)$. From
Theorem~\ref{t:BNE-char} we get the following expression for
$\bid(\val)$ where $\alloci$ is the expected allocation to agent $i$
in expectation over $\valsmi$.

\[\bid(\vali) = \vali \alloci(\vali) -
  \int_{z=0}^{\vali}\alloci(z)\,dz\]

Because the equilibrium is symmetric, one of the agents with the
highest bid is the agent with the highest value\footnote{Note that the
  bid function need only be weakly increasing, so there may be ties
  for the highest bid.}, i.e., with $\vali = \valith[1]$. We attribute
the maximum payment received by the mechanism to this agent. We can
now use the above formulation of the bid function to calculate the
expected contribution of agent $i$ to the maximum payment objective.
\begin{align*}
\MPi{\A}
&= \int_{\vali} \bid(\vali)\prob[\valsmi]{\vali=\valith[1]}\dens(\vali)\,d\vali\\
&= \int_{\vali} \left[\vali \alloci(\vali) -
  \int_{z=0}^{\vali}\alloci(z)\,dz\right]\dist(\vali)^{n-1}\densi(\vali)\,d\vali
\end{align*}
In order to simplify the second term in the integral we interchange
the order of integration over $z$ and $\vali$, integrate over $\vali$,
and then rename $z$ as $\vali$. We get:
\begin{align*} 
\MPi{\A}
&= \int_{\vali} \vali \alloci(\vali) \dist(\vali)^{n-1}\densi(\vali)\,d\vali
- \int_{\vali}\alloci(\vali) \left(\frac{1-\dist(\vali)^n}{n}\right)\,d\vali\\
&= \int_{\vali} \left\{ \vali \dist(\vali)^{n-1}
- \frac{1-\dist(\vali)^n}{n\dens(\vali)} \right\}
\times\,\alloci(\vali) \densi(\vali)\,d\vali\\
&= \int_{\vali} \alloci(\vali) \virta_n(\vali)\densi(\vali)\,d\vali\\
&= \expect[\vals]{\alloci(\vals)\virta_n(\vali)}
\end{align*}
Summing over $i$ implies the lemma.
\end{proof}

\paragraph{Optimal allocation rules and regularity.}
The characterization of Lemma~\ref{lem:mp=vv} immediately implies that
in order to maximize the expected maximum payment, we should maximize
the virtual surplus of the mechanism for maximum payment. In other
words, we should allocate the entire reward to the agent that has the
maximum virtual value $\virta_n(\vali)$ (subject to this value being
non-negative). However, this results in a monotone allocation function
only if the virtual value function is monotone non-decreasing. To this
end, we define regularity for maximum payment as follows.
\begin{definition}
A distribution $\dist$ is said to be {\em $n$-regular with respect to
  maximum payment} if $\virta_n(\cdot)$ is a monotone non-decreasing
function. The distribution is said to be {\em regular $\wrt$~maximum
  payment} if $\virta_n(\cdot)$ is monotone non-decreasing for all
positive integers $n$.
\end{definition}
For distributions that are regular $\wrt$~maximum payment, allocating
to the agent with the highest non-negative virtual value is monotone
and therefore can be implemented in BNE.  Since agents have
$\iid$ values, this outcome corresponds to allocating to the agent
with the highest value, who is in turn the agent with the highest
bid. Therefore, the optimal mechanism is a highest-bid-wins
reserve-price mechanism. The reserve value for the mechanism is given
by $\virta_n^{-1}(0)$ and the reserve bid can be computed by applying
Lemma~\ref{l:all-pay-reserves} to this value. We note that generally
the reserve price is a function of $n$ and decreases with $n$, even
for distributions that are regular for all $n$.
\begin{theorem}
\label{thm:opt-reg}
Let $\dist$ be a distribution that is $n$-regular $\wrt$ maximum
payment. Then the optimal all-pay auction for $n$ agents with values
distributed independently according to $\dist$ is a highest-bid-wins
auction with a reserve price.
\end{theorem}

\paragraph{Two examples.} We now revisit the
example with $n$ agents and values distributed according to $U[0,1]$
that was discussed in Section~\ref{sec:poa}. The following expression
defines the virtual value for maximum payment in this case:
\[\virta_n(z) = z^n(1+1/n)-1/n\,\,\text{ for } z\in [0,1]\]
This is an increasing function for all $n$. Therefore, the $U[0,1]$
distribution is regular. The optimal reserve value is given by
$\virta_n^{-1}(0) = (n+1)^{-1/n}$, and the optimal reserve bid is
$1/(n+1)$. Therefore, the optimal all-pay auction serves the highest
bidder subject to her bid being at least $1/(n+1)$. The expected
maximum payment of this auction can be calculated to be
$\frac{n}{2(n+1)}$ which approaches $1/2$ as $n$ increases.

Next consider a setting with two agents and values distributed
$\iid$ according to the exponential distribution. That is,
$\dist(\val) = 1-e^{-\val}$ for $\val\ge 0$. We can calculate the
virtual value function as $\virta_2(z) = (z-1) + e^{-z}(1/2-z)$. This
function is negative below $1.21$ and positive
thereafter. Furthermore, it is non-decreasing above $0.24$,
particularly throughout the range where it is non-negative. So
although the exponential distribution is not regular $\wrt$ maximum
payment, the optimal all-pay auction still turns out to be a
highest-bid-wins auction with a reserve price of $1.21$ and a
corresponding reserve bid of $0.85$.

An interesting point to note about the above example is that
distributions that are regular with respect to the usual notion of
virtual value for revenue, are not necessarily regular with respect to
maximum payment even for $n=2$. However, for a large subset of such
distributions, namely those that satisfy the monotone hazard rate
condition (Definition~\ref{def:mhr} below), the optimal all-pay
auction continues to have the simple form given in
Theorem~\ref{thm:opt-reg}.

\paragraph{Regularity and MHR.}
A common assumption in mechanism design literature is that
value distributions satisfy the {\em monotone hazard rate} (MHR)
condition defined below. Many common distributions such as the
uniform, Gaussian, and exponential distributions satisfy
this property. Distributions that satisfy MHR are regular and
therefore do not require ironing in the context of revenue
maximization. As our example above shows, MHR distributions are not
necessarily regular with respect to maximum payment.

\begin{definition}
\label{def:mhr}
The {\em hazard rate} of a distribution $\dist$ with density function
$\dens$ is defined as $h(x) = \frac{\dens(x)}{1-\dist(x)}$. A
distribution is said to have a {\em monotone hazard rate} (MHR) if the
hazard rate function is monotone non-decreasing.
\end{definition}

\begin{lemma}\label{lem:MHR}
Let $\dist$ be a distribution satisfying the MHR condition. Then for
any $n$ and any interval of values over which $\virta_n$ is
non-negative, $\virta_n$ is monotone non-decreasing.
\end{lemma}
\begin{proof}
We can rewrite the virtual value function in terms of the hazard rate
$h(z)$ of the distribution as follows.
\begin{align*}
\virta_n(z) &= z\dist(z)^{n-1} - \frac{1}{nh(z)}\sum_{j=0}^{n-1}\dist(z)^j\\
&= \dist(z)^{n-1}\left\{ z - \frac{1}{nh(z)}\sum_{j=0}^{n-1}\dist(z)^{-j} \right\}
\end{align*}
The function $h(z)$ is a non-negative non-decreasing
function. Therefore, $(-1/nh(z))$ is a negative non-decreasing
function. On the other hand, $\sum_{j=0}^{n-1}\dist(z)^{-j}$ is a
decreasing function of $z$. The product of a negative non-decreasing
function and a decreasing function is a non-decreasing
function. Therefore, the term within brackets is a non-decreasing
function of $z$. The term outside brackets, $\dist(z)^{n-1}$, is also
an always positive increasing function. Therefore, the product of the
two terms is an increasing function over any interval where it is
positive.
\end{proof}

\noindent
We obtain the following corollary.
\begin{corollary}
Let $\dist$ be a distribution that satisfies MHR. Then the optimal
all-pay auction for values distributed independently according to
$\dist$ is a highest-bid-wins auction with a reserve price.
\end{corollary}

\paragraph{Irregular distributions and ironing.}
For distributions that are not regular according to the definition
above, we can apply an ironing procedure from
Theorem~\ref{t:irregular-ivv-optimal} to $\virta_n$ to obtain an
ironed virtual value function $\ivirta_n$.  This function is monotone
non-decreasing and by Theorem~\ref{t:irregular-ivv-optimal} the BNE
that optimizes it point-wise optimizes the maximum payment objective.

The optimal mechanism in this case allocates the entire reward to the
agent with the maximum ironed virtual value, in the case of ties
distributing the reward equally among the tied agents\footnote{An
  equivalent way of resolving ties in the maximum ironed virtual value
  is to allocate the reward to a random tied agent.}. Since the ironed
virtual value function is a weakly increasing function, the induced
bid function is constant in the intervals where the ironed virtual
value is constant, and discontinuous at the ends of those
intervals. In effect, this creates intervals of bids that are
suboptimal to make at any value; call these bid intervals
``forbidden''. In order to implement the mechanism as an all-pay
auction, we identify the forbidden bid intervals; then we round every bid
in a forbidden bid interval down to the closest ``allowed'' bid, and
distribute the reward equally among the highest bidders (subject to an
appropriate reserve price defined by $(\ivirta_n)^{-1}(0)$). We
therefore get the following theorem:
\begin{theorem}
For any setting with $\iid$ values, the optimal all-pay auction is
defined by a reserve price and a subset of bids called forbidden bid intervals,
that has the following format: the auction solicits bids and rounds
them down to the nearest non-forbidden bids; it then distributes the
reward equally among the highest bidders subject to the bids being
above the reserve price.
\end{theorem}

\paragraph{An example of ironing.} We now present a simple example of
a distribution that is irregular $\wrt$ maximum payment, and derive
its ironed virtual value and as well as forbidden bid intervals. There are two
agents, each with a value drawn independently from $U[1,2]$ with
probability $3/4$ and from $U[2,3]$ with probability $1/4$.
  Figure~\ref{fig:ironing} below shows the virtual value function
  $\virta_2$ and its integral with respect to $q=\dist(\val)$ using
  thick grey lines; their ironed counterparts are shown in thin red
  lines. The integral of the virtual value function as a function of
  $q$ is given by the expression $\frac 12\dist^{-1}(q)(1-q^2)$. We
  iron this function by taking its convex envelope; $\ivirta_2$ is
  then the derivative with respect to $q$ of that convex envelope.

The ironed virtual value is constant in the interval $[1.918,
  2.167]$. The probability of allocation (not plotted), and therefore
the bid function, are also constant over this interval. The
corresponding bid function is plotted with a thin black line below;
there are two forbidden bid intervals, namely $[1.10, 1.199)$ and
  $(1.199, 1.31]$, with the intermediate value of $1.199$ being
allowed. The two forbidden bid intervals correspond to the two
discontinuities in the probability of allocation at the end points of
the ironed interval.

\begin{figure}
\centering
\includegraphics[scale=0.15]{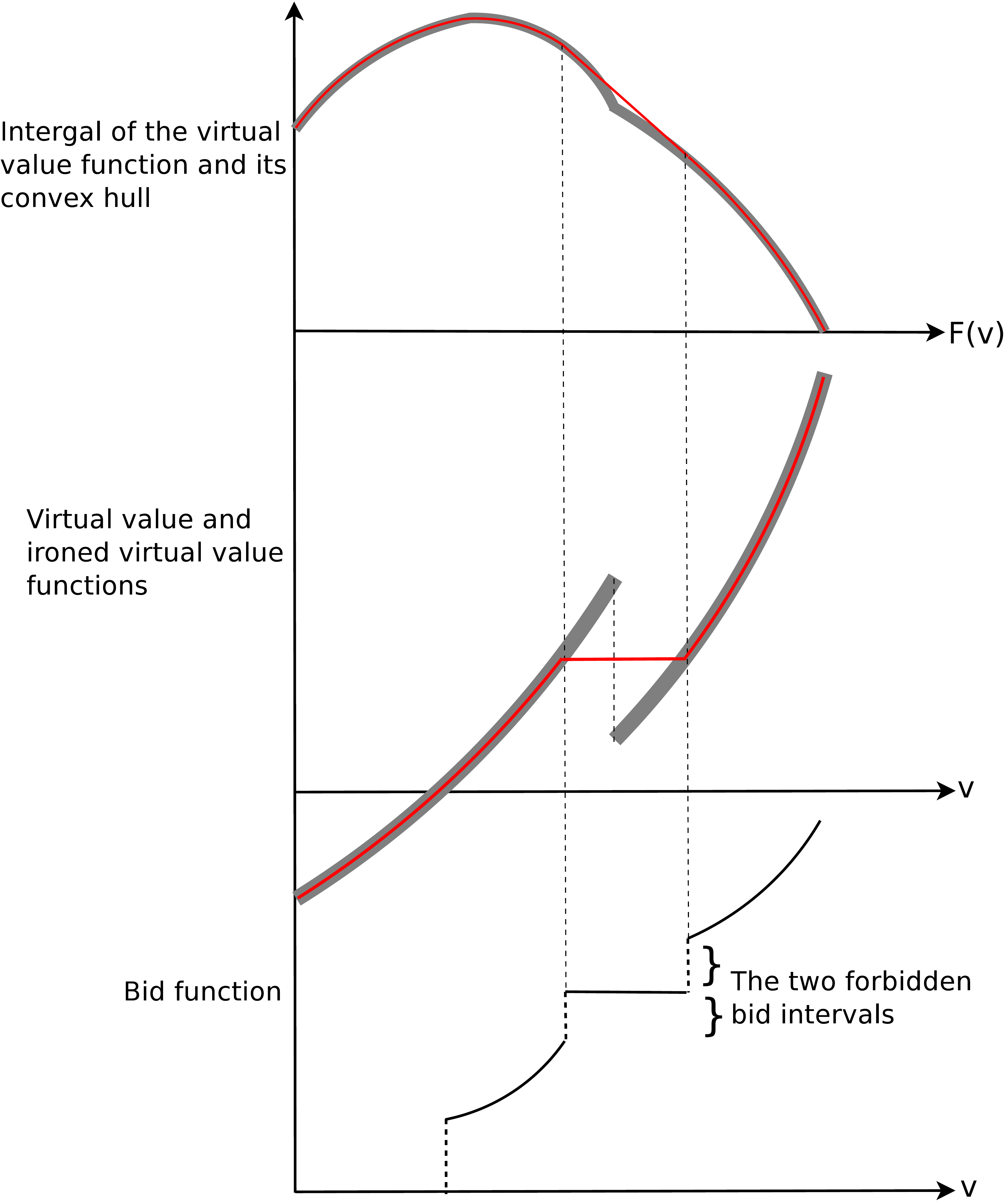}
\caption{The Ironing Procedure}
\label{fig:ironing}
\end{figure}

\paragraph{Irregularity as a function of $n$.}
An interesting point to note is that irregularity increases with
$n$. Specifically, the intervals of values that require ironing under
$\virta_n$ increase with $n$.\footnote{This happens because the
  intervals requiring ironing are precisely those where the integral
  of the virtual value function is non-concave; Increasing $n$ amounts
  to multiplying the integral with a convex function resulting in
  non-concave intervals continuing to stay non-concave.} This does not
necessarily imply that as $n$ increases a larger and larger number of
agents are tied for the reward, for two reasons: (1) reserve value
(not the reserve bid) could increase with $n$, and (2), due to the form of
the virtual value function, ironing is typically necessary at low
values rather than at high values.

\paragraph{Asymmetric contests.}
We remark that even for symmetric instances (i.e.~\iid~values)
asymmetric all-pay auctions can be more powerful than symmetric
all-pay auctions. We now present an example that exhibits this. 
Consider two contestants with values drawn \iid~ from $[0,1]$
according to the distribution $\dist(x) =
x^{1.5}$. The optimal symmetric auction studied in this paper sets a
reserve value of $(0.25)^{1/3} = 0.63$ (which translates to a
reserve bid of $(0.25)^{5/6} = 0.315$) and serves the highest
bidder who exceeds this reserve bid.  This gives an expected maximum
payment of $0.396$. 

We now define a better auction that favors contestant 1 over
contestant 2. The rules of the contest in value space are as
follows. When contestant 1's value is more than $0.75$ we serve him
irrespective of contestant 2's value, otherwise we serve the
contestant with the higher value subject to a reserve value of
$0.63$. This allocation rule creates a discontinuous increase in the
expected allocation probability of player 1 at $0.75$, and hence a
discontinuity in his bid function at $0.75$. In bid space, this
corresponds to the following contest:
\begin{enumerate}
\item We set a reserve bid of $0.315$ as before.
\item All bids of contestant $1$ in the range $[0.418, 0.681)$ get rounded down
to $0.418$.
\item When contestant 1 bids at least $0.681$ he wins irrespective of 2's 
bid.
\item Otherwise, the highest bidder wins with ties broken in favor of
  contestant 2.
\end{enumerate}

By guaranteeing victory for contestant 1 beyond a certain bid, the
above auction encourages contestant 1 to bid higher, thus boosting
maximum payment. Since the objective is maximum payment, this type of
bias is useful: the asymmetric auction obtains a smaller revenue but a
larger expected maximum payment of $0.397$. We remark that in a
real-world setting with a priori identical agents, favoring one agent
over another may be socially unacceptable.

\section{Prior-independent approximation}

As we show above, optimal crowdsourcing contests depend on knowing the
agents' value distribution. To what extent is it important to know the
distribution? In particular, under what conditions does the simple
highest-bidder-wins contest without any reserve bid approximate the
optimal one? We now show that for distributions that are regular $\wrt$
revenue the simple highest-bidder-wins contest obtains an approximation ratio
of $2n/(n-1)$, thus limiting the power of distributional knowledge.

For the standard goal of maximizing expected revenue, Bulow and Klemperer showed that for $\iid$ value
distributions that are regular $\wrt$ revenue, it is better to run a
Vickrey auction with no reserve price on $n+1$ agents than to run an
optimal auction on only $n$ agents. That is, the ability to recruit an
extra agent in the auction is more profitable to the auctioneer than
knowing the distribution.

We first note that Bulow and Klemperer's result implies that for
distributions that are regular $\wrt$ revenue, the highest-value-wins
auction with no reserve price on $n$ agents is within a factor of
$(1-1/n)$ of the optimal mechanism in terms of revenue. This combined
with Theorem~\ref{thm:util-ratio} gives us the following theorem.

\begin{theorem}
For $\iid$ distributions that are regular $\wrt$ revenue, the
highest-bid-wins all-pay auction without a reserve bid obtains an
approximation ratio of $2n/(n-1)$.
\end{theorem}

We remark that for the highest-value-wins auction without reserve
prices, the revenue converges to the optimal as more and more agents
are added. However for all-pay auctions adding more and more agents
does not improve the approximation ratio beyond $2$.

\bibliographystyle{alpha}
\bibliography{ap}

\begin{thebibliography}{CHMS10}

\bibitem[AS09]{AS09}
Nikolay Archak and Arun Sundararajan.
\newblock Optimal design of crowdsourcing contests.
\newblock In {\em International Conference on Information Systems (ICIS)},
  2009.

\bibitem[BK96]{BK-96}
J.~Bulow and P.~Klemperer.
\newblock Auctions versus negotiations.
\newblock {\em American Economic Review}, 86:180--194, 1996.

\bibitem[BKdV96]{BKdV96}
Michael~R. Baye, Dan Kovenock, and Casper~G. de~Vries.
\newblock The all-pay auction with complete information.
\newblock {\em Economic Theory}, 8:291--305, 1996.
\newblock 10.1007/BF01211819.

\bibitem[CHMS10]{CHMS-10}
Shuchi Chawla, Jason~D. Hartline, David~L. Malec, and Balasubramanian Sivan.
\newblock Multi-parameter mechanism design and sequential posted pricing.
\newblock In {\em Proceedings of the 42nd ACM symposium on Theory of
  computing}, STOC '10, pages 311--320, New York, NY, USA, 2010. ACM.

\bibitem[DV09]{dPV09}
Dominic DiPalantino and Milan Vojnovic.
\newblock Crowdsourcing and all-pay auctions.
\newblock In {\em Proceedings of the 10th ACM conference on Electronic
  commerce}, EC '09, pages 119--128, New York, NY, USA, 2009. ACM.

\bibitem[HR08]{HR-08}
Jason~D. Hartline and Tim Roughgarden.
\newblock Optimal mechanism design and money burning.
\newblock In {\em Proceedings of the 40th annual ACM symposium on Theory of
  computing}, STOC '08, pages 75--84, New York, NY, USA, 2008. ACM.

\bibitem[Min11]{M11}
Dylan Minor.
\newblock Increasing efforts through rewarding the best less.
\newblock Mansucript, 2011.

\bibitem[MS01]{MS01}
Benny Moldovanu and Aner Sela.
\newblock The optimal allocation of prizes in contests.
\newblock {\em American Economic Review}, 91(3):542--558, 2001.

\bibitem[MS06]{MS06}
Benny Moldovanu and Aner Sela.
\newblock Contest architecture.
\newblock {\em Journal of Economic Theory}, 126(1):70--97, 2006.

\bibitem[Mye81]{mye-81}
R.~Myerson.
\newblock Optimal auction design.
\newblock {\em Mathematics of Operations Research}, 6:58--73, 1981.

\bibitem[YAA08]{YAA08}
Jiang Yang, Lada~A. Adamic, and Mark~S. Ackerman.
\newblock Crowdsourcing and knowledge sharing: strategic user behavior on
  taskcn.
\newblock In {\em Proceedings of the 9th ACM conference on Electronic
  commerce}, EC '08, pages 246--255, New York, NY, USA, 2008. ACM.

\end{thebibliography}

\appendix
\section{Proof of Theorem~\ref{thm:statOpt}}\label{sec:app}

\begin{proofof}{Theorem~\ref{thm:statOpt}} We begin by noting that
the class of static auctions is symmetric, i.e., a permutation of bids results in
the same permutation of the allocation and payments. Since agents'
private values are identically distributed, any such symmetric
allocation rule induces a symmetric equilibrium in which all agents
use an identical bidding function. This in turn implies that the
allocation as a function of agents' values is also symmetric across
agents.

Let the agent values be distributed independently according to distribution function $\dist$, 
with density function $\dens$.  Consider the static allocation rule $\A =(\stat_1, \dots, \stat_k, 0,
\dots, 0)$, i.e, the agent with the $i$-th highest bid gets $\stat_i$ fraction
of the reward if $i \leq k$, an 0 otherwise.  We have $\sum_{i=1}^{k} \stat_i =
1$. We focus on the symmetric bid-function $b(\cdot)$ induced by this
allocation rule. 

In a truthful auction with allocation rule $\A$, the expected payment made by
the $r$-th highest bidder is $p_r(z) = \sum_{j=r+1}^{k+1} v_{jr}(z)(\stat_{j-1}
- \stat_j)$, where $v_{jr}(z)$ is the expectation of the $j$-th highest bid
(=value) given the $r$-th highest bid is $z$. 

Let $g(j,n,z)$ denote the expectation of the $j$-th highest draw among $n$
draws from $F$, given that the maximum draw is at most $z$. Then we have
$v_{jr}(z) = g(j-r, n-r, z)$. 

The contribution of bidder $i$ to the maximum payment objective is
\begin{align*} 
\MPi{\A} &= \int_{\vali}\bid(\vali)\prob[\valsmi]{\vali=\valith[1]}\dens(\vali)\,d\vali\\ 
&= \int_{\vali} b(\vali)\dist(\vali)^{n-1}\densi(\vali)\, d\vali 
\end{align*}
Since agents values are drawn $\iid$ from $\dist$, we have $\MP{\A} =
n\MPi{\A}$. 

Because the bid functions are symmetric, by the revenue equivalence principle, $b(z)$
equals the expected payment made by an agent with value $z$ in a truthful auction with the 
same allocation rule. So,
\begin{align*}
b(z) &= \sum_{r=1}^{k}\prob[\valsmi]{z= \valith[r]}\cdot p_r(z)\\
&= \sum_{r=1}^{k}\binom{n-1}{r-1}(1-F(z))^{r-1}F(z)^{n-r}\cdot\left\{\sum_{j=1}^{k+1-r}g(j,n-r,z)(\stat_{j+r-1} - \stat_{j+r})\right\}
\end{align*}

We prove the theorem by showing that $\frac{d\MPi{\A}}{d\stat_k}$ is negative. When we change $\stat_k$
we assume that all the mass is transferred to (or drawn from) $\stat_1$. This will prove that
the optimal allocation rule is to put all the mass on $\stat_1$, i.e., $\stat_1 = 1$. 

Using the formula for $b(z)$,
it is easy to observe that for $r=2$ to $r=k-1$, terms corresponding to that specific $r$ 
in $\frac{d\MPi{\A}}{d\stat_k}$ will be an integral with an integrand
of
$$\binom{n-1}{r-1}(1-F(z))^{r-1}F(z)^{2n-r-1} \cdot\left\{-g(k-r,n-r,z) + g(k-r+1, n-r, z) \right\}$$ 
This integrand is negative because $g$ is a decreasing function in its first argument. 

The term corresponding to $r=1$ in $\frac{d\MPi{\A}}{d\stat_k}$ will be an integral with an integrand of
$$F(z)^{2n-2} \cdot\left\{-g(1,n-1,z)- g(k-1,n-1,z) + g(k,n-1,z)\right\} $$
Note that the above integrand is negative even if $g(1,n-1,z)$ term were not there. 

The term corresponding to $r=k$ in $\frac{d\MPi{\A}}{d\stat_k}$ will be an integral with a positive integrand of 
$$ \binom{n-1}{k-1}(1-F(z))^{k-1}F(z)^{2n-k-1}\cdot\left\{ g(1,n-k,z) \right\}$$

Our proof is going to upper bound  $\frac{d\MPi{\A}}{d\stat_k}$ 
by ignoring certain negative terms in it, and show that even the upper bound is negative.
In particular, we only consider terms corresponding to $r=k-1$, $r=k$ and one term of 
$r=1$, namely $F(z)^{2n-2} \cdot\left\{-g(1,n-1,z)\right\}$. Let this upper bound be
denoted by $Q$. 
\begin{align*}
\frac{d\MPi{\A}}{d\stat_k} \leq Q &= - \int_{z} F(z)^{2n-2} g(1,n-1,z)\,d\dist(z)\\
&- \binom{n-1}{k-2} \int_z(1-F(z))^{k-2} F(z)^{2n-k} g(1,n-k+1,z)\,d\dist(z)\\
&+ \binom{n-1}{k-2} \int_z(1-F(z))^{k-2} F(z)^{2n-k} g(2,n-k+1,z)\,d\dist(z)\\
&+ \binom{n-1}{k-1} \int_z(1-F(z))^{k-1} F(z)^{2n-k-1} g(1,n-k,z)\,d\dist(z)
\end{align*}
We derive the expressions for $g(1,n,z)$ and $g(2,n,z)$ below. 
\begin{align*}
g(1,n,z) &= n\int_0^{z} y\frac{f(y)}{F(z)}\left(\frac{F(y)}{F(z)}\right)^{n-1}\,dy\\
&= z - \frac{\int_0^z F(t)^n\,dt}{F(z)^n}
\end{align*}
\vspace{-1cm}
\begin{align*}
g(2,n,z) &= n(n-1)\int_0^{z} y\frac{f(y)}{F(z)}\left(1 - \frac{F(y)}{F(z)}\right)\left(\frac{F(y)}{F(z)}\right)^{n-2}\,dy\\
&= z - \left[n\frac{\int_0^z F(t)^{n-1}\,dt}{F(z)^{n-1}} - (n-1)\frac{\int_0^z F(t)^n\,dt}{F(z)^n}\right]
\end{align*}
We susbtitute the expression for $g$ into $Q$. 
\begin{align*}
Q &= -\int_z F(z)^{n-2}\left[zF(z)^n - \int_zF(t)^n \,dt\right]\,d\dist(z)\\
&+ \binom{n-1}{k-1}\int_z (1-\dist(z))^{k-1}F(z)^{2n-k-1}z\,d \dist(z)\\
& + \binom{n-1}{k-2}\int_z (1-\dist(z))^{k-2}F(z)^{n-1} \left(\int_0^z F(t)^{n-k+1}\, dt\right)\, d\dist(z)\\
& + \binom{n-1}{k-2}(n-k)\int_z (1-\dist(z))^{k-2}F(z)^{n-1} \left(\int_0^z F(t)^{n-k+1}\, dt\right)\, d\dist(z)\\
& - \binom{n-1}{k-2}(n-k+1)\int_z (1-\dist(z))^{k-2}F(z)^{n} \left(\int_0^z F(t)^{n-k}\, dt\right)\, d\dist(z)\\
& - \binom{n-1}{k-1}\int_z (1-\dist(z))^{k-1}F(z)^{n-1} \left(\int_0^z F(t)^{n-k}\, dt\right)\, d\dist(z)
\end{align*}
We now factor the term $(1-\dist(z))^{k-1}$ as $(1-\dist(z))^{k-2}\cdot(1-\dist(z))$ and then group terms. We get
\begin{align*}
Q &= -\int_z F(z)^{n-2}\left[zF(z)^n - \int_zF(t)^n \,dt\right]\,d\dist(z)\\
&- \binom{n-1}{k-1}\int_z (1-\dist(z))^{k-2}F(z)^{2n-k}z\,d \dist(z)\\
&+ \binom{n-1}{k-1}\int_z (1-\dist(z))^{k-2}F(z)^{2n-k-1}z\,d \dist(z)\\
&+  \binom{n-1}{k-2}(n-k+1)\int_z (1-\dist(z))^{k-2}F(z)^{n-1} \left(\int_0^z F(t)^{n-k+1}\, dt\right)\, d\dist(z)\\
&-  \binom{n-1}{k-2}\left[(n-k+1) - \frac{n-k+1}{k-1}\right] \int_z (1-\dist(z))^{k-2}F(z)^{n} \left(\int_0^z F(t)^{n-k}\, dt\right)\, d\dist(z)\\
&-  \binom{n-1}{k-1}\int_z (1-\dist(z))^{k-2}F(z)^{n-1} \left(\int_0^z F(t)^{n-k}\, dt\right)\, d\dist(z)
\end{align*}
We have to prove that $Q\leq 0$. This is equivalent to proving that
\begin{align*}
&\int_z F(z)^{n-2}\left[zF(z)^n - \int_0^zF(t)^n \,dt\right]\,d\dist(z)\\
&+ \binom{n-1}{k-1}\int_z (1-\dist(z))^{k-2}F(z)^{n-1}\left[zF(z)^{n-k+1} -\int_0^z F(t)^{n-k+1}\,dt \right]\,d\dist(z)\\
&- \binom{n-1}{k-1}\int_z (1-\dist(z))^{k-2}F(z)^{n-1}\left[zF(z)^{n-k} -\int_0^z F(t)^{n-k}\,dt \right]\,d\dist(z)\\
&\geq\\
& \binom{n-1}{k-1}(k-2)\int_z (1-\dist(z))^{k-2}F(z)^{n-1}\left(\int_0^z F(t)^{n-k+1}\,dt \right)\,d\dist(z)\\
&- \binom{n-1}{k-1}(k-2)\int_z (1-\dist(z))^{k-2}F(z)^{n}\left(\int_0^z F(t)^{n-k}\,dt \right)\,d\dist(z)
\end{align*}

The RHS can be seen to be negative. Thus it is enough to prove that the LHS is positive. Rewriting the terms
in the square bracket via integration by parts, 
\begin{align*}
&n\int_z \dist(z)^{n-2}\left(\int_0^z tF(t)^{n-1}\,d \dist(t)\right)\, d\dist(z)\\
&+ \binom{n-1}{k-1}(n-k+1) \int_z (1-\dist(z))^{k-2}\dist(z)^{n-1}\left(\int_0^z tF(t)^{n-k}\,d \dist(t)\right)\, d\dist(z)\\
&- \binom{n-1}{k-1}(n-k) \int_z (1-\dist(z))^{k-2}F(z)^{n-1}\left(\int_0^z tF(t)^{n-k-1}\,d \dist(t)\right)\, d\dist(z)
\end{align*}
Changing the order of integration, we have the LHS as,
\begin{equation*}
\int_{t=0}^{\infty}t \dist(t)^{n-k-1}\dens(t)\left\{
\begin{split}
\binom{n-1}{k-1}(n-k+1)\left(\int_{\dist(t)}^{1}(1-\dist(z))^{k-2}\dist(z)^{n-1}\,d\dist(z)\right)
\left[\dist(t) - \frac{n-k}{n-k+1}\right]\qquad \\
\qquad + n\left(\int_{\dist(t)}^1 \dist(z)^{n-2}\,d \dist(z)\right)\dist(t)^k
\end{split}
\right\}\, dt
\end{equation*}
Applying integration by parts again, (this time taking $t$ as one term and the rest as the 
differential part) we get the LHS as,
\begin{equation*}
\int_{t=0}^{\infty} \left( \int_{\dist(t)}^1 u^{n-k-1}\left\{
\begin{split}
\binom{n-1}{k-1}(n-k+1)\left(\int_{u}^{1}(1-\dist(z))^{k-2}\dist(z)^{n-1}\,d\dist(z)\right)
\left[u - \frac{n-k}{n-k+1}\right]\qquad \\
\qquad + n\left(\int_{u}^1 \dist(z)^{n-2}\,d \dist(z)\right)u^k
\end{split}
\right\}\,du\right)\, dt
\end{equation*}
Rewrite the above integral as 
$\int_{t=0}^{\infty} H_n(\dist(t))\,dt$
where 
$$ H_n(x) = \int_{x}^1 u^{n-k-1}\left\{
\binom{n-1}{k-1}(n-k+1)\left(\int_{u}^{1}(1-v)^{k-2}v^{n-1}\,dv\right)
\left[u - \frac{n-k}{n-k+1}\right] + n\left(\int_{u}^1 v^{n-2}\,dv\right)u^k\right\}\,du
$$
If we prove that $H_n(x)$ is always non-negative for $x \in [0,1]$ we are done. We have
$$-H_n'(x)= x^{n-k-1}\left\{
\binom{n-1}{k-1}(n-k+1)\left(\int_{x}^{1}(1-v)^{k-2}v^{n-1}\,dv\right)
\left[x - \frac{n-k}{n-k+1}\right] + n\left(\int_{x}^1 v^{n-2}\,dv\right)x^k\right\}
$$

Observe that $-H_n'(x)$ is negative for small values of $x$
and positive for large values of $x$ and never becomes negative after it has become
positive. Thus, $H_n(x)$ is first increasing and then decreasing. 
We know that $H_n(1) = 0$. If we prove that $H_n(0) \geq 0$, we would have proven that $H_n(x)$
is always non-negative. 
\begin{align*}
H_n(0) &= \binom{n-1}{k-1}(n-k+1)\int_{0}^1 u^{n-k-1}
\left(\int_{u}^{1}(1-v)^{k-2}v^{n-1}\,dv\right)
\left[u - \frac{n-k}{n-k+1}\right]\,du\\
&+ n\int_{0}^1\left(\int_{u}^1 v^{n-2}\,dv\right) u^{n-1}\,du\\
&=\binom{n-1}{k-1}(n-k+1)\int_{0}^1 (1-v)^{k-2}v^{n-1} 
\left(\int_{0}^{v}u^{n-k-1}\left[u - \frac{n-k}{n-k+1}\right]\,du\right)\,dv\\
&+ n\int_{0}^1 v^{n-2}\left(\int_{0}^v u^{n-1}\,du\right)\,dv\\
&=\binom{n-1}{k-1}\int_{0}^1 (1-v)^{k-2}v^{2n-k-1}(v-1)\,dv + \frac{1}{2n-1}\\
&=-2\binom{n-1}{k-1}\int_{0}^{\pi/2}cos^{4n-2k-1}(\theta)sin^{2k-1}(\theta)\,d\theta + \frac{1}{2n-1}
\end{align*}

The integral $\int_{0}^{\pi/2}cos^m(\theta)sin^{n}(\theta)\,d \theta = \frac{\Gamma(\frac{m+1}{2})\Gamma(\frac{n+1}{2})}{2\Gamma(\frac{m+n+2}{2})}$
Accordingly, we have 
\begin{align*}
H_n(0) &= - \binom{n-1}{k-1}\frac{\Gamma(2n-k)\Gamma(k)}{\Gamma(2n)} + \frac{1}{2n-1}\\
& > 0
\end{align*}
\end{proofof}

\end{document}